%% file: main.tex
\documentclass[letterpaper, 10 pt, conference]{ieeeconf}  

\IEEEoverridecommandlockouts                              

\usepackage{generic}
\usepackage{amsmath,amssymb,amsfonts}
\usepackage{graphicx}
\usepackage{textcomp}
\usepackage[english]{babel}
\usepackage[utf8x]{inputenc}
\usepackage[colorinlistoftodos]{todonotes}
\usepackage{comment}
\usepackage{algorithm}
\usepackage[noend]{algpseudocode}
\usepackage[normalem]{ulem}
\renewcommand{\baselinestretch}{0.96}
\long\def\comment#1{}

\newif\ifpreprint
\preprinttrue

\input{styles}

\title{\LARGE \bf
State dimension reduction of recurrent equilibrium networks\\
with contraction and robustness preservation
}
\author{M.F. Shakib
\thanks{M.F. Shakib is with the Department of Electrical and Electronic Engineering, Imperial College London, London, UK {\tt (m.shakib@imperial.ac.uk)}.}%
}

\begin{document}

\maketitle
\thispagestyle{empty}
\pagestyle{empty}

\begin{abstract}
Recurrent equilibrium networks (RENs) are effective for learning the dynamics of complex dynamical systems with certified contraction and robustness properties through unconstrained learning.
While this opens the door to learning large-scale RENs, deploying such large-scale RENs in real-time applications on resource-limited devices remains challenging.
Since a REN consists of a feedback interconnection of linear time-invariant (LTI) dynamics and static activation functions, this article proposes a projection-based approach to reduce the state dimension of the LTI component of a trained REN.
One of the two projection matrices is dedicated to preserving contraction and robustness by leveraging the already-learned REN contraction certificate.
The other projection matrix is iteratively updated to improve the accuracy of the reduced-order REN based on necessary $h_2$-optimality conditions for LTI model reduction.
Numerical examples validate the approach, demonstrating significant state dimension reduction with limited accuracy loss while preserving contraction and robustness.
\end{abstract}
\vspace{-3mm}
\input{S1_introduction}

\input{S2_problem_setting}

\input{S3_main_result}

\vspace{1mm}
\input{S4_numerical_example}

\vspace{1mm}

\section{CONCLUSIONS}\label{sec:conclusions}
Recurrent equilibrium networks (RENs) offer a robust framework for learning dynamical models but can be too computationally demanding for real-time use. 
This article proposes a projection-based order reduction method that preserves the REN's contraction and robustness properties during the reduction process. 
The method is based on a novel projection strategy leveraging the already learned contraction certificate and an iterative algorithm that updates the reduced-order REN to improve its accuracy.
Numerical experiments demonstrate significant order reduction, the preservation of contraction and robustness, and only little loss in accuracy.
Future research includes complexity reduction by reducing the number of neurons.

\ifpreprint{\input{A_proof}}\fi


\balance

\bibliographystyle{ieeetr}
\bibliography{MyBIB}             

\end{document}

%% file: styles.tex
\usepackage{graphicx}          
\usepackage{amsmath}
\usepackage{amssymb}
\usepackage{algorithm}
\usepackage[noend]{algpseudocode}
\newtheorem{theorem}{Theorem}
\newtheorem{assumption}{Assumption}

\usepackage{xcolor}
\usepackage{caption}
\usepackage{subcaption}
\usepackage[scr=boondoxo, scrscaled=.98]{mathalfa} 
\usepackage{setspace}
\usepackage{soul}
\usepackage{multirow}
\usepackage{hhline}
\usepackage{booktabs}
\usepackage{float}
\usepackage{mathtools}
\usepackage{hyperref}
\usepackage{balance}
\newtheorem{definition}{Definition}
\newtheorem{lemma}{Lemma}
\newtheorem{remark}{Remark}
\newcommand{\VV}{\mathbb{V}}
\newcommand{\WW}{\mathbb{W}}
\newcommand{\RR}{\mathbb{R}}
\newcommand{\CC}{\mathbb{C}}

\newcommand{\NN}{\mathbb{N}}

\newcommand{\JJ}{\mathcal{J}}

\newcommand{\LL}{\mathcal{L}}
\newcommand{\UU}{\mathcal{U}}

\newcommand{\CCC}{\mathcal{C}}

\newcommand{\RRR}{\mathcal{R}}

\newcommand{\Ht}{h_2}

\def\NoNumber#1{{\def\alglinenumber##1{}\State #1}\addtocounter{ALG@line}{-1}}
\newcommand{\norm}[1]{\left\lVert#1\right\rVert}

\renewcommand{\algorithmicrequire}{\textbf{Input:}}
\renewcommand{\algorithmicensure}{\textbf{Output:}}
\algnewcommand\algorithmicswitch{\textbf{switch}}
\algnewcommand\algorithmiccase{\textbf{case}}
\algnewcommand\algorithmicassert{\texttt{assert}}
\algnewcommand\Assert[1]{\State \algorithmicassert(#1)}%
\algdef{SE}[SWITCH]{Switch}{EndSwitch}[1]{\algorithmicswitch\ #1\ \algorithmicdo}{\algorithmicend\ \algorithmicswitch}%
\algdef{SE}[CASE]{Case}{EndCase}[1]{\algorithmiccase\ #1}{\algorithmicend\ \algorithmiccase}%
\algtext*{EndSwitch}%
\algtext*{EndCase}%

\usepackage{generic}
\usepackage{cite}
\usepackage{textcomp}
\usepackage{soul}
\def\BibTeX{{\rm B\kern-.05em{\sc i\kern-.025em b}\kern-.08em
    T\kern-.1667em\lower.7ex\hbox{E}\kern-.125emX}}

%% file: S1_introduction.tex
\section{INTRODUCTION}

The increasing complexity of engineering systems has led to a corresponding rise in the complexity of dynamical models used for their design, control, and operation. 
This trend has fuelled the development of data-driven models that employ \emph{artificial neural networks} to learn underlying system dynamics~\cite{ljung2020deep,beintema2023deep}. 
However, such models often suffer from excessive over-parameterizations~\cite{pillonetto2025deep}, both in terms of the number of states and network complexity, while lacking robustness guarantees for deviations from training data. 
These limitations present significant challenges related to memory efficiency, computational tractability, and generalizability.

\emph{Recurrent Equilibrium Networks} (RENs)~\cite{revay2023recurrent} provide an expressive framework for learning dynamical models. 
RENs generalize deep, recurrent, and convolutional neural networks while inherently ensuring contraction~\cite{lohmiller1998contraction} and robustness, such as, e.g., a bounded incremental $\ell_2$-gain, as certified through \emph{incremental integral quadratic constraints} (IQC)~\cite{megretski2002system}. 
These robustness guarantees are certified by a \emph{contraction metric certificate}, which is learned as an auxiliary parameter during the training phase. 
Thanks to these properties, RENs have gained traction in diverse applications, including system identification\cite{revay2023recurrent} and control policy synthesis~\cite{revay2023recurrent,Wang2023Learning}. 
Unlike constrained optimization approaches~\cite{tobenkin2017,shakib2022computationally,shakib2023kernel}, RENs enable \emph{unconstrained} training through \emph{direct parametrizations}~\cite{revay2023recurrent}, while still guaranteeing contraction and robustness, making them well-suited for learning large-scale dynamical models with built-in contraction and robustness guarantees.

\begin{figure}
    \centering
    \includegraphics[width=0.95\linewidth]{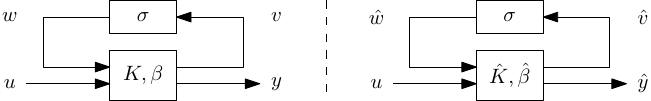}
    \caption{Left: a large-scale REN~\eqref{eq:system} with the LTI block~\eqref{eq:system:LTI} parameterized by the weights $K$ and biases $\beta$, and $\sigma$ the nonlinear block~\eqref{eq:system:NL}. Right: A reduced-order REN with weights $\hat W$ and biases $\hat \beta$, and the same nonlinear block~$\sigma$.}
    \label{fig:REN_schematic}
    \ifpreprint{}
    \else{\vspace{-0.3cm}}\fi
\end{figure}

While such training enables scalable learning of large-scale RENs, these models can be intractable for real-time use. 
Given a trained REN, this article proposes an \emph{order reduction} framework that \emph{preserves} the contraction and robustness. 
Since a REN is an interconnection of LTI dynamics and static nonlinearities, we focus on projection-based reduction of the LTI part, see Fig.~\ref{fig:REN_schematic}, which uses two projection matrices~\cite{antoulas2005approximation}. 
Our first contribution is a novel selection strategy that preserves contraction and robustness: one projection matrix is free, the other is fixed from the pre-learned contraction metric. 
This idea is inspired by moment-matching reduction for Lur'e-type systems~\cite{shakib2023model,Shakib2024optimal} and recent LTI methods~\cite{cheng2023optimal,shakib2025dissipativity}, but avoids large-scale semi-definite programs by exploiting the existing REN certificate.
The proposed approach is also applicable to the recently introduced R2DN class~\cite{barbara2025r2dnscalableparameterizationcontracting}, which also enjoys built-in contraction and robustness properties.

The second contribution is an iterative algorithm that updates the free projection matrix to improve the accuracy of the reduced-order REN. 
Since the accuracy of the reduced-order REN is inherently tied to the accuracy of its reduced-order LTI component, the algorithm uses the $h_2$-optimal LTI reduction ideas~\cite{gugercin2008h_2,bunse2010h2}, akin to IRKA~\cite{gugercin2008h_2}, and related extensions~\cite{xu2011optimal,bunse2010h2}.
In line with these existing methods, our iterative algorithm has no formal convergence guarantee.
However, unlike them, it preserves the contraction and robustness properties of the reduced REN throughout the iterations.
Numerical results confirm that the proposed method substantially reduces state dimension while maintaining accuracy and certified contraction and robustness.

The main contributions of this article are as follows:
\begin{itemize}
    \item \textbf{Contraction- and robustness-preserving projection:} a framework that preserves these properties of the large-scale REN for the reduced-order REN.
    \item \textbf{Iterative algorithm:} a scalable algorithm that combines LTI $\Ht$-optimal model reduction with the contraction- and robustness-preserving projection framework.
    \item \textbf{Numerical validation:} Examples demonstrating substantial state dimension reduction while preserving contraction, robustness, and accuracy.
\end{itemize}

The remainder of this paper is organized as follows. Section~\ref{sec:preliminaries} introduces the preliminaries on RENs.
Section~\ref{sec:main} presents the main theoretical contributions.
Section~\ref{sec:example} provides numerical examples.
Section~\ref{sec:conclusions} presents the main conclusions and potential future research directions.
\ifpreprint{}
\else{
The proofs are omitted for brevity and can be found in~\cite{shakib2025state}.}
\fi

\textbf{Notation:}
The symbols $\RR, \CC$, and $\NN$ denote the set of real, complex, and natural numbers, respectively.
For~$A \in \RR^{n \times m}$, $A^\top$ denotes its transpose and col$(A)$ denotes its column space, while for $A \in \CC^{n \times m}$, $A^*$ denotes its Hermitian transpose.
The notation $\text{blkdiag}(A_1,A_2)$ denotes a block-diagonal matrix of submatrices $A_1$ and $A_2$.
For a symmetric matrix $A \in \RR^{n \times n}$, $A \succ 0$ $(A \prec 0)$ denotes that $A$ is positive (negative) definite, and $A \succeq 0$ $(A \preceq 0)$ denotes that $A$ is positive (negative) semi-definite.
The set of positive-definite diagonal matrices is denoted by~\( \mathbb{D}_+ \).
The set~\( \ell^n_{2e} \) denotes the set of sequences~\( x : \mathbb{N} \to \mathbb{R}^n \), where \( x_t \in \mathbb{R}^n \) is the value of the sequence at time \( t \in \mathbb{N} \) and we omit $n$ when the dimension is clear from the context.
The subset~\( \ell_2 \subset \ell_{2e} \) consists of all sequences that satisfy $\|x\| < +\infty$ with $\|x\| \coloneq \sqrt{\sum_{t=0}^{\infty} |x_t|^2}$ and \( |(\cdot)| \) the Euclidean norm.
The norm of a truncated sequence~\( x \in \ell_{2e} \) over \([0, T]\) is defined by $\|x\|_T := \sqrt{\sum_{t=0}^{T} |x_t|^2}.$
For two sequences~\( x, y \in \ell^n_{2e} \), the inner product over \([0, T]\) is defined as~$\langle x, y \rangle_T \coloneqq \sum_{t=0}^{T} x_t^T y_t$.
A set is closed under conjugation if it contains the complex conjugate of each of its elements.

%% file: S2_problem_setting.tex
\section{Recurrent equilibrium networks}\label{sec:preliminaries}
This section first recalls the REN model class in Section~\ref{S:systemclass}.
After that, the notions of \emph{contraction} and \emph{robustness} are recalled in Section~\ref{S:CR-REN}.
Finally, conditions for characterizing these notions are outlined in Section~\ref{S:conditions_for_CR_RENs}.

\subsection{Model class}\label{S:systemclass}
A discrete-time recurrent equilibrium network is represented by the following dynamical system:
\begin{subequations}
\label{eq:system}
\begin{align}\label{eq:system:LTI}
\begin{bmatrix}
    x_{t+1} \\
    v_t \\
    y_t
\end{bmatrix}
&=
\overbrace{
\begin{bmatrix}
    A & B_1 & B_2 \\
    C_1 & D_{11} & D_{12} \\
    C_2 & D_{21} & D_{22}
\end{bmatrix}}^{K\in \RR^{(n+q+p)\times(n+q+m)}}
\begin{bmatrix}
    x_t \\
    w_t \\
    u_t
\end{bmatrix}
+\!\!\!\!
\overbrace{
\begin{bmatrix}
    \beta_x \\
    \beta_v \\
    \beta_y
\end{bmatrix}}^{\beta\in\RR^{n+q+p}}
\!\!\!\!,
\\ \label{eq:system:NL}
    w_t &= \sigma(v_t) := 
    \begin{bmatrix}
        \sigma(v_t^1) & \sigma(v_t^2) & \dots & \sigma(v_t^q)
    \end{bmatrix}^{\top},
\end{align}
\end{subequations}
where, at time $t\in\NN$, $x_t\in\RR^n$ is the state, $u_t\in\RR^m$ is the input, $y_t\in\RR^p$ is the output, and $v_t,w_t\in\RR^{q}$ are internal signals.
The state dimension is denoted by $n$ and the number of neurons is denoted by $q$. 
The matrix $K \in \RR^{(n+q+p)\times(n+q+m)}$ contains the \emph{weights} of the REN with its sub-matrices $A,\ldots, D_{22}$ having appropriate dimensions, while $\beta\in\RR^{n+q+p}$ contains the \emph{biases} of the REN with $\beta_x \in \RR^n, \beta_v\in\RR^q$, and $\beta_y \in \RR^p$.
The internal signals~$v$ and~$w$ are connected through the nonlinear activation function~$\sigma$, see~\eqref{eq:system:NL}, which operates element-wise on~$v$.

A REN~\eqref{eq:system} is a feedback interconnection between the LTI block~\eqref{eq:system:LTI} (where~$\beta$ is seen as an external input to this LTI block) and the activation functions~\eqref{eq:system:NL}, see Fig.~\ref{fig:REN_schematic}.
To assess the contraction and robustness of the REN, conditions on the block (given in the next section) and the nonlinear block (summarized next) can be assessed separately.

\begin{assumption}\label{ass:incr_secetor_condition}
    The nonlinearity $\sigma$ satisfies the 
    following incremental
    condition:
    \begin{equation}\label{eq:incr_sector_condition}
        0 \le \frac{\sigma(y)-\sigma(x)}{y-x} \le 1, \forall (x,y)\in\RR, x\neq y.
    \end{equation}
\end{assumption}
\vspace{1mm}

Many of the commonly used activation functions, such as, e.g., the ReLu and tanh functions, satisfy Assumption~\ref{ass:incr_secetor_condition}~\cite{revay2023recurrent}.

The matrix $D_{11}$ in~\eqref{eq:system} gives rise to the implicit equation\footnote{A REN is well-posed if the implicit equation $v = C_1x+D_{11}\sigma(v) + D_{12}u+\beta_v$ has a unique solution $v\in\RR^q$ for any $x\in\RR^n$ and $u\in\RR^m$.} $v = C_1x + D_{11}\sigma(v) + D_{12}u + \beta_v$ in $v$, enhances the representation capabilities of the REN by introducing depth and recurrent behavior~\cite{revay2023recurrent}. 
A special case arises when $D_{11}$ is lower triangular, reducing the network to a feedforward network~\cite{revay2023recurrent}. 
In this case, the $i$-th element of $v$ depends only on its preceding elements~$i-1,i-2,\ldots$, allowing for a row-by-row computation and eliminating the implicit nature. 
Nevertheless, the results presented in this article hold for both lower-triangular and full matrices $D_{11}$.

\subsection{Contracting and robust RENs} \label{S:CR-REN}

This article focuses on contracting and robust RENs. 
The notion of \emph{contraction} is recalled first from~\cite{revay2023recurrent}.

\vspace{1mm}
\begin{definition}[\protect{Contraction~\cite[Definition~2]{revay2023recurrent}}]
\label{def:contraction}
    A REN~\eqref{eq:system} is said to be contracting with rate $\alpha \in (0,1)$ if for any two initial conditions $x_0^a, x_0^b \in \RR^n$, given the same input sequence $u \in \ell_{2e}^m$, the state sequences $x^a$ and $x^b$ satisfy $\left|x_t^a - x_t^b\right| \le \tau \alpha^t \left|x_0^a -x_0^b \right|,$
    for some $\tau > 0$.\hfill$\square$
\end{definition}
\vspace{1mm}

A contracting system `forgets' its initial condition as all solutions with different initial conditions contract to each other.
A REN can also exhibit \emph{robustness} in terms of the sensitivity of the REN solutions with respect to changes in the input and the initial conditions via incremental IQCs.
To introduce this concept, denote by $\RRR_{x_0}(u)$ the output sequence that corresponds to the initial condition $x_0\in\RR^n$ and the input $u  \in \ell_{2e}^m$ such that~\eqref{eq:system} is~satisfied. 

\begin{definition}\label{def:robustness}
[\protect{Robustness~\cite[Definition~3]{revay2023recurrent}}]
A REN~\eqref{eq:system} is said to satisfy the \textit{incremental IQC} defined by \((Q, S, R)\), where \( 0 \succeq Q \in \mathbb{R}^{p \times p} \), \( S \in \mathbb{R}^{m \times p} \), and \( R = R^\top \in \mathbb{R}^{m \times m} \), if for all pairs of solutions with initial conditions \( a, b \in \mathbb{R}^n \) and input sequences \( u^a, u^b \in \ell^m_{2e} \), the output sequences \( y^a = \RRR_a(u^a) \) and \( y^b = \RRR_b(u^b) \) satisfy
\begin{equation}\label{eq:IQC_ineq}
\hspace{-1mm}
    \sum_{t=0}^{T}
    \begin{bmatrix}
        y_t^a - y_t^b \\
        u^a_t - u^b_t
    \end{bmatrix}^\top
    \begin{bmatrix}
        Q & S^\top \\
        S & R
    \end{bmatrix}
    \begin{bmatrix}
        y_t^a - y_t^b \\
        u^a_t - u^b_t
    \end{bmatrix}
    \geq -d(a, b), \, \forall \, T,
\end{equation}
for some function \( d(a, b) \geq 0 \) with \( d(a, a) = 0 \).\hfill$\square$
\end{definition}

Next, we recall conditions for contraction and robustness.

\subsection{Conditions for contraction and robustness} \label{S:conditions_for_CR_RENs}
The characterizations of contraction (Definitions~\ref{def:contraction}) and robustness (Definition~\ref{def:robustness}), are based on Assumption~\ref{ass:incr_secetor_condition} and matrix inequality conditions presented next.

\vspace{1mm}
\begin{theorem}[~\!\!\!\!\protect{\cite[Theorem~1]{revay2023recurrent}}]\label{thm:LMI_cond}
Consider the REN~\eqref{eq:system} satisfying Assumption~\ref{ass:incr_secetor_condition} and a given \(\bar{\alpha} \in (0,1]\).

\begin{enumerate}
    \item \textbf{Contracting REN:} Suppose there exists \(P = P^{\top} \succ 0\) and \(\Lambda \in \mathbb{D}_{+}\) such that
    \begin{equation}\label{eq:cond:contraction}
    \begin{bmatrix}
    \bar{\alpha}^{2}P & -C_{1}^{\top}\Lambda \\
    -\Lambda C_{1} & Y
    \end{bmatrix}
    -
    \begin{bmatrix}
    A^{\top} \\
    B_{1}^{\top}
    \end{bmatrix}
    P
    \begin{bmatrix}
    A^{\top} \\
    B_{1}^{\top}
    \end{bmatrix}^{\top}
    \succ 0,
    \end{equation}
    where \(Y \coloneqq 2\Lambda - \Lambda D_{11} - D_{11}^{\top}\Lambda\). Then, the REN is well-posed and contracting with some rate \(\alpha < \bar{\alpha}\).

    \item \textbf{Robust REN:} Consider the incremental IQC defined in~\eqref{eq:IQC_ineq} with \((Q, S, R)\) given, where \(Q \preceq 0\). Suppose there exist \(P = P^{\top} \succ 0\) and \(\Lambda \in \mathbb{D}_{+}\) such that
    \begin{multline}\label{eq:cond:IQC}
    \!\!\!\!\!\!\!\!
    \begin{bmatrix}
    \bar{\alpha}^{2}P & -C_{1}^{\top}\Lambda & C_{2}^{\top}S^{\top} \\
    -\Lambda C_{1} & Y & D_{21}^{\top}S^{\top} - \Lambda D_{12} \\
    SC_{2} & SD_{21} - D_{12}^{\top}\Lambda & R + SD_{22} + D_{22}^{\top}S^{\top}
    \end{bmatrix}
    \\
    -
    \begin{bmatrix}
    A^{\top} \\
    B_{1}^{\top} \\
    B_{2}^{\top}
    \end{bmatrix}
    P
    \begin{bmatrix}
    A^{\top} \\
    B_{1}^{\top} \\
    B_{2}^{\top}
    \end{bmatrix}^{\top}
    +
    \begin{bmatrix}
    C_{2}^{\top} \\
    D_{21}^{\top} \\
    D_{22}^{\top}
    \end{bmatrix}
    Q
    \begin{bmatrix}
    C_{2}^{\top} \\
    D_{21}^{\top} \\
    D_{22}^{\top}
    \end{bmatrix}^{\top}
    \succ 0.
    \end{multline}
    Then, the REN is well-posed, satisfies~\eqref{eq:IQC_ineq}, and is contracting with some rate $\alpha < \bar \alpha$. \hfill $\blacksquare$
\end{enumerate}
\end{theorem}
\medskip

For RENs with large~$n$ and/or~$q$, training a contracting or robust REN while solving these matrix inequalities during the training process is often intractable due to the involved computational complexity.
To mitigate this issue,~\cite{revay2023recurrent} proposed a \emph{direct parameterization}.
Loosely speaking, a direct parameterization is a mapping from $\RR^{n_\theta}$, for some $n_\theta \in \NN$, to the weights $K$ and biases $\beta$ of a REN~\eqref{eq:system} such that for any choice in $\RR^{n_\theta}$, the corresponding REN is contracting and/or robust with a given $(Q,S,R)$.
As an auxiliary variable, the \emph{certificate} $P$ is also trained in this approach.

\begin{remark}
    Contraction is intimately connected to \emph{incremental stability} and \emph{convergence}, see~\cite{jungers2024discrete}.
    In fact, under the conditions in Theorem~\ref{thm:LMI_cond}, the REN is also incrementally stable and globally exponentially convergent, see~\cite{jungers2024discrete,tran2018convergence}.
\end{remark}

\begin{remark}
    For continuous-time simplified RENs with $D_{11}=0$,~\cite{besselink2011model,shakib2023model,Shakib2024optimal} exploit (incremental) conditions, similar to those in Theorem~\ref{thm:LMI_cond}, to provide error bounds on the reduction error.
    Herein, enforcing the reduced-order model to satisfy the same incremental conditions has proven vital.
\end{remark}

%% file: S3_main_result.tex
\section{State dimension reduction for RENs}\label{sec:main}

Given a contracting or robust REN of order~$n$ satisfying Theorem~\ref{thm:LMI_cond}, the goal is to construct a reduced-order REN of order $\hat n < n$ that preserves contraction/robustness for the \emph{same} $(Q,S,R)$ and yields responses  `close' to the responses of the original $n$-dimensional REN.
The proposed reduction approach focuses only on reducing the state dimension of the LTI dynamics only, see also Fig.~\ref{fig:REN_schematic}.
The first main result, presented in Section~\ref{S:3:proj_REN}, shows that the REN contraction certificate $P$ enables order reduction to any chosen reduced order by fixing one projection matrix.
This is inherently tied to the accuracy of the REN as it consists of an interconnection between LTI dynamics and static activation functions, and only the LTI component is reduced.
The other projection matrix is then optimized to reduce the $\Ht$-norm of the error between the LTI components of the full and reduced RENs.
A formal analysis of how this $\Ht$ error relates to output trajectory errors is left for future work.
We recall $\Ht$-optimal model reduction for LTI systems in Section~\ref{S:3:proj_LTI} and propose an iterative algorithm in Section~\ref{S:3:iterative_alg} that updates the free projection matrix to reduce this error while preserving contraction and robustness at every iteration.

\subsection{A robustness-preserving projection for RENs}\label{S:3:proj_REN}
Consider the following \emph{reduced-order} REN:
\begin{align}\label{eq:ROM}
\begin{bmatrix}
    \hat x_{t+1} \\
    \hat v_t \\
    \hat y_t
\end{bmatrix}
=
\WW^\top
K
\VV
\begin{bmatrix}
    \hat x_t \\
    \hat w_t \\
    u_t
\end{bmatrix}
+
\WW^\top \beta,
\quad
    \hat w_t = \sigma(\hat v_t),
\end{align}
where the projection matrices $\WW\in\CC^{(n+q+p)\times(\hat n + q + p)}$ and $\VV^{(n+q+m)\times(\hat n + q + m)}$ have the following structure
\vspace{-1mm}
\begin{equation}\label{eq:def_VW_bold}
\! \! \WW^\top \! \coloneqq \text{blkdiag}(W^\top, I_q, I_p), \quad 
\VV \coloneqq \text{blkdiag}(V, I_q, I_m),
\end{equation}
and where $W,V\in\RR^{n\times \hat n}$ are such that $W^\top V = I$.
Given any projection matrix $V$ that is full column rank, let
\vspace{-1mm}
\begin{equation}\label{eq:def_W}
    W^\top \coloneqq \left(V^\top P V \right)^{-1} V^\top P,
    \vspace{-1mm}
\end{equation}
where $0 \prec P\in\RR^{n\times n}$ is the certificate in Theorem~\ref{thm:LMI_cond}, either for contraction and/or robustness.
The following result shows that the reduced-order REN~\eqref{eq:ROM} is contraction- and robustness-preserving with these projection matrices.
\vspace{1mm}

\begin{theorem}\label{thm:projection_REN}
    Consider system~\eqref{eq:system} and let $V\in\RR^{n \times \hat n}$ be a matrix such that rank$(V) = \hat n$.
    \begin{enumerate}
        \item Suppose that system~\eqref{eq:system} satisfies all the conditions of Theorem~\ref{thm:LMI_cond} for \textit{contraction} with certificate $P$ and rate $\bar \alpha$.
        Then, the reduced-order REN~\eqref{eq:ROM} with $W$ as in~\eqref{eq:def_W} also satisfies the conditions of Theorem~\ref{thm:LMI_cond} for \textit{contraction} with the same rate $\bar \alpha$.
        \item Suppose that system~\eqref{eq:system} satisfies all the conditions of Theorem~\ref{thm:LMI_cond} for \textit{robustness} for a given triple $(Q,S,R)$ with certificate $P$ and rate $\bar \alpha$.
        Then, the reduced-order REN~\eqref{eq:ROM} with $W$ as in~\eqref{eq:def_W} also satisfies the conditions of Theorem~\ref{thm:LMI_cond} for \textit{robustness} for the same $(Q,S,R)$ and rate $\bar \alpha$.
        Hence, the reduced-order REN~\eqref{eq:ROM} is contractive and robust for the same $(Q,S,R)$.
        \ifpreprint{}
        \else{
        \hfill $\blacksquare$}
        \fi
    \end{enumerate}
\end{theorem}
\ifpreprint{
\begin{proof}
    The proof and an auxiliary lemma can be found in the appendix.
\end{proof}
}\fi

This theorem allows for an arbitrarily small state dimension~$\hat n$ of the reduced-order REN, while still preserving contraction and/or robustness.
Using numerical examples, the role of the order $\hat n$ in the accuracy of the reduced REN is investigated in Section~\ref{sec:example}.
The selection of~$V$ to obtain an accurate reduced-order REN is addressed in the next section.

\subsection{$\Ht$-optimal model reduction for LTI systems}
\label{S:3:proj_LTI}
The reduction error in the LTI component directly influences the error between the large-scale and the reduced REN.
To address this, $h_2$-optimal reduction for LTI systems is~revisited.
Consider a discrete-time LTI system described~by
\vspace{-1mm}
\begin{equation}\label{eq:sys_LTI}
    x_{k+1} = A x_k + B u_k, \quad y_k = C x_k + Du_k,
    \vspace{-1mm}
\end{equation}
where, at time $k \in \NN$, $x_k\in \RR^n$, $u_k\in\RR^m$, and $y_k\in\RR^p$, and $A\in\RR^{n\times n},B\in\RR^{n\times m},C\in\RR^{p\times n},$ and $D\in\RR^{p\times m}$.
Similarly, consider the reduced-order model described by
\vspace{-1mm}
\begin{equation}\label{eq:ROM_LTI}
    \hat x_{k+1} = \hat A \hat x_k + \hat B u_k, \quad \hat y_k = \hat C \hat x_k + \hat Du_k,
\end{equation}
where, $\hat x_k\in \RR^{\hat n}$, $u_k\in\RR^m$, and $\hat y_k\in\RR^p$, and $\hat A\in\RR^{\hat n\times \hat n},\hat B\in\RR^{\hat n\times m},\hat C\in\RR^{p\times \hat n}$, and $\hat D\in\RR^{p\times m}$.

For a stable discrete-time transfer function $H:\CC\rightarrow\CC^{p\times m}$, the $\Ht$-norm is defined as
$\norm{H}_{\Ht}^2 \coloneqq \frac{1}{2\pi} \int_{-\infty}^{+\infty} \text{trace}\left[ H(e^{j\omega})^\star H(e^{j\omega})\right]\text{d}\omega,$
leading to the $\Ht$-norm of the error dynamics induced by reduction:
\vspace{-1mm}
\begin{equation}\label{eq:def_JJ}
    \JJ \coloneqq \norm{H(z)-\hat H(z)}_{\Ht}^2,
\end{equation}
where $H(z)\coloneqq C(zI-A)^{-1}B+D$ is the transfer function of the original system~\eqref{eq:sys_LTI} and $\hat H(z)\coloneqq \hat C(zI-\hat A)^{-1}\hat B+\hat D$ of the reduced-order model~\eqref{eq:ROM_LTI}.
Necessary conditions~\eqref{eq:ROM_LTI} to be a minimizer of $\JJ$ are recalled from~\cite{bunse2010h2}.

\begin{theorem}[~\!\!\!\!\cite{bunse2010h2}]\label{thm:optimality_conditions}
    Given the system~\eqref{eq:sys_LTI} with simple poles and transfer function \( H(z) \). 
    Let \( \hat{H}(z) \) be the transfer function of the reduced order model~\eqref{eq:ROM_LTI} with simple poles.
    Let~\eqref{eq:ROM_LTI} be in an eigenvector basis \( \hat{A} = \operatorname{diag} \left( \hat{\lambda}_1, \dots, \hat{\lambda}_{\hat n} \right) \), \( \hat{B} = \begin{bmatrix} \hat{b}_1^\star, \dots, \hat{b}_{\hat n}^\star \end{bmatrix}^\star \), \( \hat{C} = \begin{bmatrix} \hat{c}_1, \dots, \hat{c}_{\hat n} \end{bmatrix} \), and $\hat D = D$. 
    If \( \hat{H}(z) \) is a minimizer of $\JJ$ in~\eqref{eq:def_JJ}, then the following conditions are satisfied for $k = 1, \dots, \hat n$:
\begin{align*}
\textstyle{
\hat{c}_k^\star H \left( \frac{1}{\hat{\lambda}_k^{\star}} \right) =
\hat{c}_k^\star \hat H \left( \frac{1}{\hat{\lambda}_k^{\star}} \right), \quad 
H \left( \frac{1}{\hat{\lambda}_k^{\star}}  \right) \hat{b}_k^\star = \hat{H} \left( \frac{1}{\hat{\lambda}_k^{\star}} \right) \hat{b}_k^\star,} \\
\textstyle{
\hat{c}_k^\star H' \left( \frac{1}{\hat{\lambda}_k^{\star}} \right) \hat{b}_k^\star = \hat{c}_k^\star \hat{H}' \left( \frac{1}{\hat{\lambda}_k^{\star}} \right) \hat{b}_k^\star, \qquad \qquad \quad \!\!
}%
\end{align*}
where \( (1/{\hat{\lambda}_k^{\star}})  \) are the mirrored images of the poles of~\eqref{eq:ROM_LTI} with respect to the unit circle, \( \hat{b}_k \) is the \( k \)-th row of \( \hat{B} \), and \( \hat{c}_k \) is the \( k \)-th column of \( \hat{C} \), and where $H'(\bar z)\coloneqq \left.\text{d} H(z)/\text{d}z \right|_{z=\bar z}$ is the first derivative of $H$ evaluated at $\bar z\in\CC$.
\hfill$\blacksquare$
\end{theorem}

An extended version of this result covers the case of repeated poles in the system or reduced-order model~\cite{bunse2010h2}.
The necessary conditions in Theorem~\ref{thm:optimality_conditions} are stated in terms of the tangential interpolation data~$\{ \hat\lambda_k,\hat c_k , \hat b_k \}_{k=1}^{\hat{n}}$ obtained from the \emph{optimal} reduced-order model~\eqref{eq:ROM_LTI}.
A projection-based reduction technique for LTI systems that provides a reduced-order model that achieves interpolation at a given set of tangential interpolation data is recalled next.

\begin{lemma}[~\!\!\!\!\cite{bunse2010h2}]\label{lem:tangential_interpolation}
Let \( V , W \in \RR^{n \times \hat n} \) be rank~\( \hat n \) matrices such that \( W^\top V = I_{\hat n} \). 
Let \( \tau_k \in \mathbb{C} \), \( l_k \in \mathbb{C}^{1 \times p} \), and \( r_k \in \mathbb{C}^{m \times 1} \) for \( k = 1, \dots, \hat n \), be given sets of
tangential interpolation data that are closed under conjugation and such that \( A - \tau_k I_n \) is invertible.
If for all \( k = 1, \dots, \hat n \),
\begin{align*}
    (\tau_k I - A)^{-1} B r_k \in \text{col}(V), \quad \!\!
    (\tau_k I - A^\star)^{-1} C^\star l_k^\star \in \text{col}(W),
\end{align*}
then~\eqref{eq:ROM_LTI}~with~$(\hat{A}, \hat{B}, \hat{C}, \hat D) = (W^\top A V,\allowbreak W^\top B, C V,D)$ has a transfer function which satisfies for $k = 1, \dots, \hat n$:
\begin{align*}
     l_k H \left( \tau_k \right) &= l_k \hat{H} \left( \tau_k \right), \quad 
    H \left( \tau_k \right) r_k = \hat{H} \left( \tau_k \right) r_k, \\
    &\ell_k H' \left( \tau_k \right) r_k = \ell_k \hat{H}' \left( \tau_k \right) r_k. 
    \tag*{\text{$\blacksquare$}}
\end{align*} 
\end{lemma}
\vspace{1mm}

Lemma~\ref{lem:tangential_interpolation} can be used to construct a reduced-order model that satisfies \emph{given} interpolation data~$\{ \tau_k ,l_k ,r_k \}_{k=1}^{\hat{n}}$.
However, after constructing such a reduced-order model, it typically does not satisfy the necessary conditions of Theorem~\ref{thm:optimality_conditions}.

To address this issue, an algorithm can be defined to iteratively update the tangential data until the reduced-order model satisfies the conditions of Theorem~\ref{thm:optimality_conditions}.
Starting from an initial choice of interpolation data, Lemma~\ref{lem:tangential_interpolation} is applied to construct a reduced-order model.
Based on Theorem~\ref{thm:optimality_conditions}, the interpolation data is then updated from the spectral information of the current reduced-order model, and the process is repeated.
This algorithm was first proposed in~\cite{gugercin2008h_2} under the acronym IRKA and was later generalized to multivariable LTI systems in~\cite{bunse2010h2,xu2011optimal}.
Although the convergence of such algorithms is not guaranteed, these algorithms often converge quickly in practice~\cite{gugercin2008h_2}, however, without guaranteeing stability preservation, even if the original LTI system was stable.

An alternative iterative algorithm for LTI systems that preserves stability was presented in~\cite{gugercin2008iterative}.
This algorithm updates only \emph{one} of the projection matrices, namely the matrix~$V$.
In particular, it relies on the following result.
\vspace{1mm}

\begin{lemma}\label{lem:tangential_interpolation_r}
Let all variables be defined as in Lemma~\ref{lem:tangential_interpolation}.
If for all \( k = 1, \dots, \hat n \),
$(\tau_k I - A)^{-1} B r_k \in \text{col}(V),$
then, for any matrix $W\in \RR^{n\times\hat n}$ such that $W^\top V = I$, the model~\eqref{eq:ROM_LTI} with \((\hat{A}, \hat{B}, \hat{C}, \hat D) = (W^\top A V, W^\top B, C V, D) \) has a transfer function which satisfies:%
\vspace{1mm}
\begin{equation}\label{eq:tangential_interpolation_r}
    H \left( \tau_k \right) r_k = \hat{H} \left( \tau_k \right) r_k,
    \vspace{-2mm}
\end{equation}
for $k = 1, \dots, \hat n.$ 
\hfill $\blacksquare$
\end{lemma}
\vspace{1mm}

It was shown in~\cite{gugercin2008iterative} that using $W$ as in~\eqref{eq:def_W} with~$P$ replaced by the observability Gramian of the original system~\eqref{eq:sys_LTI}, preserves stability for the reduced-order model at each iteration of a corresponding iterative algorithm called ISRK.
However, it should be noted that ISRK, once converged, only guarantees that the necessary conditions
\begin{align}\label{eq:necessary_condition_b}
 \textstyle{H \left( \frac{1}{\hat{\lambda}_k^{\star}}  \right) \hat{b}_k^\star = \hat{H} \left( \frac{1}{\hat{\lambda}_k^{\star}} \right) \hat{b}_k^\star, \quad \text{for } k = 1, \dots, \hat n,}
\end{align}
are satisfied, not those corresponding to~$\hat c^\star_k H(\cdot)$ and $H'(\cdot)$.
Inspired by the ISRK algorithm~\cite{gugercin2008iterative}, the next section proposes an iterative algorithm that preserves contraction and/or robustness for the reduced-order RENs.

\subsection{An iterative algorithm for the state reduction of RENs}
\label{S:3:iterative_alg}

\begin{algorithm}[t]
	\caption{An iterative contraction/robustness preserving reduction algorithm for recurrent equilibrium networks}\label{Alg}
	\algorithmicrequire { } REN of order $n$ with a certificate $P \succ 0$ for contraction and/or robustness and a reduction order $\hat n \leq n$.
    Let $B=\begin{bmatrix} B_1 & B_2 \end{bmatrix}$
	\begin{algorithmic}[1]
		\State Take an initial selection of the shifts~$\tau_i$ and the directions~$r_i$, for $i=1,\ldots,\hat n$, such that the sets $\{\tau_i\}_{i=1}^{\hat n}$ and $\{r_i\}_{i=1}^{\hat n}$ are closed under conjugation and such that $(\tau_i I - A)$ is non-singular for all $i$.
		\While{relative change in $\{\tau_i,r_i\}_{i=1}^{\hat n} >$ tolerance}
        \State Construct $V$ according to Lemma~\ref{lem:tangential_interpolation_r}.
        \vspace{0.5mm}
        \State Take $W^\top = \left(V^\top P V\right)^{-1}V^\top P$ as in~\eqref{eq:def_W}.
		\State Construct $\hat K \leftarrow \WW^\top K \VV$ with $\WW$ and $\VV$ as in~\eqref{eq:def_VW_bold}.
        \State Compute the eigenvalue decomposition 
        \vspace{1mm}
        \NoNumber{$\qquad\hat A = \hat X\hat\Omega \hat X^{-1}$, 
        where $\hat\Omega = \text{diag}\left(\hat \lambda_1,\ldots, \hat\lambda_{\hat n}\right)$.}
        \vspace{1mm}
        \State Update $\tau_1 \leftarrow \left(\hat \lambda_1^\star\right)^{-1}, \ldots, \tau_{\hat n} \leftarrow \left(\hat \lambda_{\hat n}^\star\right)^{-1}$,
        \vspace{1mm}
        \NoNumber{$\begin{bmatrix}
            r_1 & \ldots & r_{\hat n}
        \end{bmatrix} \leftarrow \left(\hat X^{-1} \begin{bmatrix}
            \hat B_1 & \hat B_2
        \end{bmatrix}\right)^\star$.}
		\EndWhile
	\end{algorithmic}
	\algorithmicensure { } A contracting or robust REN~\eqref{eq:ROM} with $\hat\beta \leftarrow \WW^\top \beta$.
\end{algorithm}

Directly applying IRKA~\cite{gugercin2008h_2} or its extensions~\cite{bunse2010h2,xu2011optimal,gugercin2008iterative} to the LTI component of a REN generally fails to preserve contraction or robustness, as these methods do not ensure stability preservation. 
Even stability-preserving variants such as ISRK~\cite{gugercin2008iterative} are insufficient for guaranteeing these properties. 
To address this, we propose an ISRK-inspired algorithm, grounded in Theorems~\ref{thm:projection_REN}–\ref{thm:optimality_conditions} and Lemma~\ref{lem:tangential_interpolation_r}, that preserves contraction and robustness for RENs.

The iterative algorithm is outlined in Algorithm~\ref{Alg}.
It takes the large-scale REN~\eqref{eq:system} with a contraction or robustness certificate \( P \) and a desired \( \hat{n} \) as inputs.
Step~1 initializes the interpolation data \( \{\tau_k, r_k\}_{k=1}^{\hat{n}} \) (e.g., taken randomly).
Step~2 consists of a while loop, where, Steps~3 and 4 update the projection matrices \( V \) and $W$ according to Lemma~\ref{lem:tangential_interpolation_r}
and~\eqref{eq:def_W}, respectively.
Step~5 performs the projection to obtain the reduced-order REN~\eqref{eq:ROM}, which satisfies~\eqref{eq:tangential_interpolation_r}, where \( H \) and \(\hat H\) denote the transfer functions of the LTI component of the large-scale REN~\eqref{eq:system} and of the reduced-order REN~\eqref{eq:ROM}, respectively.
Step~6 brings the reduced-order LTI component into the eigenvector realization and Step~7 updates the tangential data \( \{\tau_k, r_k\}_{k=1}^{\hat{n}} \).
If the relative change in the tangential data remains larger than a tolerance, the process returns to Step~3.
Upon convergence, the bias vector $\hat \beta$ of the reduced REN is updated.
The properties of Algorithm~\ref{Alg} are summarized next for a fixed point of Algorithm~\ref{Alg}, i.e., a point in which the change in the interpolation data is zero.
\vspace{1mm}

\begin{theorem}\label{thm:summary}
    Consider a given REN~\eqref{eq:system} and suppose that $A$ has distinct eigenvalues and that the conditions of Theorem~\ref{thm:LMI_cond} for contraction and/or robustness are satisfied for some $(Q,S,R)$ with certificate $P$. 
    At a fixed point of Algorithm~\ref{Alg}, if rank$(V)=\hat n$ and the eigenvalues of $\hat A$ are distinct, then:
    \begin{itemize}
        \item The necessary conditions~\eqref{eq:necessary_condition_b} for the optimality of the $h_2$-error between the LTI component of the REN~\eqref{eq:system} and the reduced-order REN~\eqref{eq:ROM} are satisfied, and
        \item The reduced-order REN~\eqref{eq:ROM} is contracting and/or robust for the same $(Q,S,R)$ with certificate $V^\top P V$. 
        \ifpreprint{}
        \else{\hfill $\blacksquare$}
        \fi
    \end{itemize}
\end{theorem}
\ifpreprint{
\begin{proof}
    The proof follows from the application of Theorems~\ref{thm:projection_REN} and~\ref{thm:optimality_conditions}, and Lemma~\ref{lem:tangential_interpolation_r}.
\end{proof}\fi
\medskip

\begin{remark}
Some remarks about Algorithm~\ref{Alg} are in place:
\begin{itemize}
    \item Connecting to~\cite{de1981controllability,ionescu2014families,shakib2023time}, the matrix $V$ in Step~3 of Algorithm~\ref{Alg} is of rank $\hat n$ if the pair $(A,B)$ is controllable and the pair $(\hat \Omega,\begin{bmatrix} r_1 & \ldots & r_{\hat n} \end{bmatrix})$ is observable.
    \item Algorithm~\ref{Alg} is computationally efficient because (i) it leverages the pre-trained certificate~\( P \) of the large-scale REN and (ii) the projection matrix~\( V \) can be efficiently computed using Lanczos/Arnoldi procedures~\cite{gugercin2008h_2}. 
\end{itemize}
\end{remark}

%% file: S4_numerical_example.tex
\section{Numerical example}
\label{sec:example}

This example considers a randomly-generated contracting and robust REN with $n=100$, $q=100$, $m=1$, and $p=1$.
The robustness parameters are $Q=-\frac{1}{\gamma}, R = \gamma$, and $S=0$, which corresponds to an incremental $\ell_2$-gain bound of $\gamma$, where $\gamma = 2$.
The goal is to reduce the REN using Algorithm~\ref{Alg} to different orders $\hat n$ ranging from $\hat n = 1$ to $\hat n = 50$.
Since Algorithm~\ref{Alg} is sensitive to the initial tangential data provided, it has been run ten times for each $\hat n$ for randomly generated initial tangential data.
The source code can be downloaded from \url{www.github.com/FahimShakib}.

Fig.~\ref{fig:iteration_history_n100} depicts the iteration history of Algorithm~\ref{Alg} for different initial tangential data for the case $\hat n =10$.
Here, the $\Ht$-error is consistently reduced within the first few iterations, and the algorithm converges to the same fixed point, satisfying the condition~\eqref{eq:necessary_condition_b}.
Since this condition is only necessary, the model with the smallest $\Ht$-error has always been backtracked and selected.
For the input depicted in the top plot of Fig.~\ref{fig:time_resp_n100}, the response~$\hat y$ of the reduced-order model is plotted in Fig.~\ref{fig:time_resp_n100} together with the response~$y$ of the large-scale REN with $n=100$.
These responses remain close to each other with the errors more than one order of magnitude smaller than the responses themselves.

\begin{figure}
    \centering
    \includegraphics[width=0.9\linewidth]{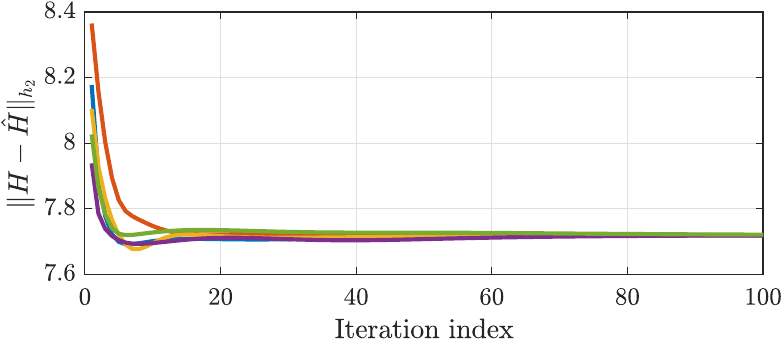}
    \caption{Typical iteration history of Algorithm~\ref{Alg}.
    Different lines correspond to different initial tangential data.}
    \vspace{-2mm}
    \label{fig:iteration_history_n100}
\end{figure}

\begin{figure}
    \centering
    \includegraphics[width=0.9\linewidth]{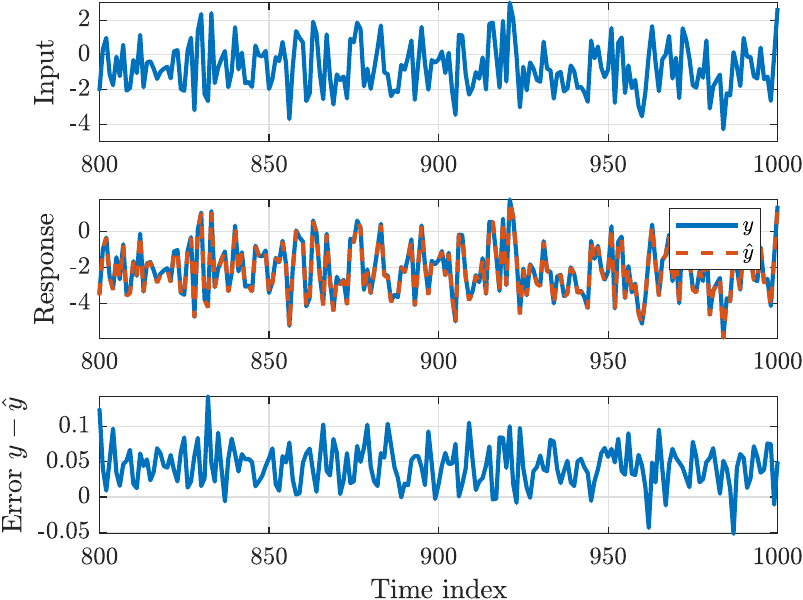}
    \caption{Top: the last 200 samples of a white noise input sequence. Middle: the response $y$ of the large-scale REN with $n=100$ and the response $\hat y$ of the reduced-order REN with $\hat n = 10$. Bottom: the error between $y$ and~$\hat y$.}
    \vspace{-0.5cm}
    \label{fig:time_resp_n100}
\end{figure}

Simulations are used to quantify the accuracy
of the reduced-order REN.
First, $n_u$ white-noise input sequences $\UU_i = \{u_t \in \RR \}_{t=1}^T, i =1,\ldots,n_u,$ are realized from normal distributions with different variances and means.
These are used consistently throughout this example for $T=1000$.
Then, for input $\UU_i$, with $i=1,\ldots,n_u$, the error measures
\begin{equation}\label{eq:def_C}
\hspace{-1mm}
\textstyle{
    \CCC_i \coloneqq \frac{\norm{\RRR_0(\UU_i)-\hat\RRR_0(\UU_i)}_T}{\norm{\RRR_0(\UU_i)}_T}\cdot 100\%, \ \  \CCC \coloneqq \frac{1}{n_u} \sum_{i=1}^{n_u} \CCC_i,
    }
\end{equation}
are computed, where $\RRR_0(\UU_i)$ is the response of the large-scale REN to input $\UU_i$ and $\hat \RRR_0(\UU_i)$ is the response of the reduced-order REN for the same input, both starting from zero initial conditions and where $\CCC$ is the average of $\{\CCC_i\}_{i=1}^{n_u}$.
For $n_u=10$, Fig.~\ref{fig:Performance_n100} depicts~$\CCC$ over~$\hat n$ together with the $h_2$-error of the LTI components of the RENs.
As expected, the $h_2$-error decreases consistently for increasing orders $\hat n$.
However, $\CCC$ does not decrease consistently which is because (i) a smaller $\Ht$-error in the LTI part does not necessarily lead to a smaller $\CCC$ as a REN is a nonlinear system, and (ii) the reduced RENs are only tested for $n_u$ inputs, while the $\Ht$-error considers any input.
In general, especially for small $\hat n$, obtaining a smaller $h_2$-error in the LTI component results in smaller errors in REN responses, typically only a few percent, even for a large reduction in order.

\begin{figure}
    \centering
    \includegraphics[width=0.9\linewidth]{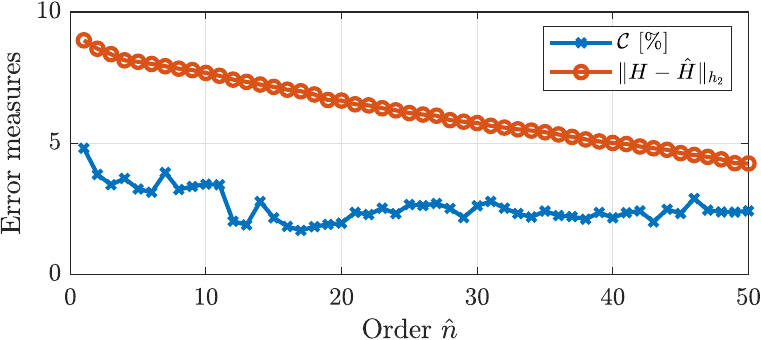}
    \caption{The error measure $\CCC$ in~\eqref{eq:def_C} and the $h_2$-error between the LTI parts of the $n$-th-order the $\hat n$-th-order RENs.}
    \vspace{-1mm}
    \label{fig:Performance_n100}
\end{figure}

Finally, the incremental $\ell_2$-gain of the reduced-order RENs is computed empirically to validate whether the bound $\gamma$ is satisfied in simulation.
Consider the following measure:
\begin{equation}\label{eq:empericalincrl2gain_large_REN}
    \textstyle{\beta (\mathcal{R}) = \max_{i,j\in \{1,\ldots,n_u\}, i\neq j}\frac{\|\mathcal{R}(\UU_i) - \mathcal{R}(\UU_j)\|_T}{\|\UU_i - \UU_j\|_T}.}
\end{equation}
Here, $\tilde \gamma \coloneqq \beta(\mathcal{R}_0)$ is a lower bound on the incremental $\ell_2$-gain of the large-scale REN, while, for each~$\hat n$, $\hat \gamma \coloneqq \beta(\hat {\mathcal{R}}_0)$ is a lower bound on the incremental $\ell_2$-gain of the reduced-order REN.
Fig.~\ref{fig:l2gain_n100} depicts $\hat \gamma, \tilde \gamma$, and $\gamma = 2$,
where the large-scale REN satisfies $\tilde \gamma \le \gamma$ with a gap between $\gamma$ and $\tilde\gamma$ because $\gamma$ is only an upper bound for the REN's actual incremental $\ell_2$-gain.
Furthermore, in line with Theorems~\ref{thm:projection_REN} and~\ref{thm:summary}, also $\hat \gamma < \gamma$ is satisfied, while it is observed that $\hat \gamma \approx \tilde \gamma$.

\begin{figure}
    \centering
    \includegraphics[width=0.9\linewidth]{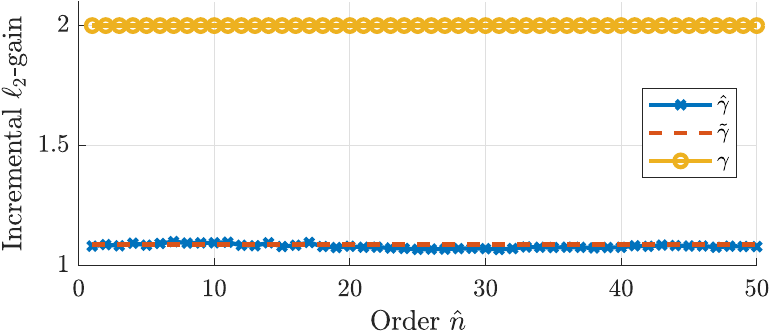}
    \caption{The estimated lower-bound $\hat \gamma$ on the incremental $\ell_2$-gain of the reduced-order RENs for various orders $\hat n$, and the analog lower-bound $\tilde \gamma$ for the large-scale REN. Both $\hat \gamma$ and $\tilde \gamma$ remain below the upper bound $\gamma=2$.}
    \vspace{-0.3cm}
    \label{fig:l2gain_n100}
\end{figure}

%% file: A_proof.tex
\section*{APPENDIX: Proof of Theorem~\ref{thm:projection_REN}}\label{A:proof_thm:projection_REN}
Before presenting the proof, an auxiliary lemma and its proof are presented.
\begin{lemma}\label{lem:P_ineq}
    For any $P\in \RR^{n\times n}$ such that $P=P^\top \succ 0$ and any $V \in \RR^{n\times \hat n}$ such that rank$(V)=\hat n$, the inequality
    \begin{equation}\label{eq:lemma_ineq}
        P - PV\left(V^\top P V\right)^{-1} V^\top P \succeq 0
    \end{equation}
    holds true. \hfill$\blacksquare$
\end{lemma}
\begin{proof}
    Apply a Schur complement to~\eqref{eq:lemma_ineq} to yield the equivalent inequality
    \begin{align*}
        \begin{bmatrix}
           V^\top P V & V^\top P \\ PV &P 
        \end{bmatrix} = 
        \begin{bmatrix}
           V^\top & 0 \\ 0 & I
        \end{bmatrix}
        \underbrace{\begin{bmatrix}
           P & P \\ P &P 
        \end{bmatrix}}_{\eqqcolon \tilde P \succeq 0}
        \underbrace{\begin{bmatrix}
           V & 0 \\ 0 & I
        \end{bmatrix}}_{\eqqcolon\tilde M}
        \succeq 0.
    \end{align*}
    Noting that $P\succ 0$ and taking a Schur complement of $\tilde P$ leads to $P-PP^{-1}P = P-P = 0 \succeq 0$, hence $\tilde P \succeq 0$.
    Then, inequality~\eqref{eq:lemma_ineq} holds since $\tilde P \succeq 0$ and $\tilde M$ is full column rank (because $V$ is full column rank). 
\end{proof}
\vspace{1mm}

Next, the proof of Theorem~\ref{thm:projection_REN} is presented.

\begin{proof}
Only the statement of item 2) (robustness and contraction) is proven. 
The statement of item 1) (contraction only) follows a similar reasoning and is omitted for brevity.

The conditions for robustness (as in~\eqref{eq:cond:IQC} in Theorem~\ref{thm:LMI_cond}) can be written compactly as follows:
\begin{equation}
    P=P^\top \succ 0 , \quad \LL_n \coloneqq F - G^\top P G + H^\top Q H \succ 0,
\end{equation}
where
\begin{align*}
    F &\coloneqq \begin{bmatrix}
    \bar{\alpha}^{2}P & -C_{1}^{\top}\Lambda & C_{2}^{\top}S^{\top} \\
    -\Lambda C_{1} & Y & D_{21}^{\top}S^{\top} - \Lambda D_{12} \\
    SC_{2} & SD_{21} - D_{12}^{\top}\Lambda & R + SD_{22} + D_{22}^{\top}S^{\top}
    \end{bmatrix}, \\
    G &\coloneqq 
    \begin{bmatrix}
    A &
    B_{1} &
    B_{2}
    \end{bmatrix},\quad 
    H \coloneqq     
    \begin{bmatrix}
    C_{2} &
    D_{21} &
    D_{22}
    \end{bmatrix}.
\end{align*}
For the reduced-order REN in~\eqref{eq:ROM}, these conditions are:
\begin{equation}\label{lem:condition_rom_REN}
    \hat P=\hat P^\top \succ 0 , \quad \hat F - \hat G^\top \hat P \hat G + \hat H^\top Q \hat H \succ 0,
\end{equation}
where
\begin{align*}
    \hat F &\coloneqq \begin{bmatrix}
    \bar{\alpha}^{2}\hat P & -V^\top C_{1}^{\top}\Lambda & V^\top C_{2}^{\top}S^{\top} \\
    -\Lambda C_{1}V & Y & D_{21}^{\top}S^{\top} - \Lambda D_{12} \\
    SC_{2}V & SD_{21} - D_{12}^{\top}\Lambda & R + SD_{22} + D_{22}^{\top}S^{\top}
    \end{bmatrix},\\
    \hat G &\coloneqq 
    W^\top\begin{bmatrix}
     A V &
     B_{1} &
     B_{2}
    \end{bmatrix}, \quad
    \hat H \coloneqq     
    \begin{bmatrix}
    C_{2}V &
    D_{21} &
    D_{22}
    \end{bmatrix}.
\end{align*}
Take $\hat P \coloneqq V^\top P V$ and note that~\eqref{lem:condition_rom_REN} can be written using
\begin{align*}
    \hat F &= M^\top F M, \quad \hat G = W^\top GM, \quad \hat H = HM,
\end{align*}
where $M\coloneqq \text{blkdiag}(V,I,I)$.
Then,~\eqref{lem:condition_rom_REN} is equivalent to
\begin{subequations}
\begin{align}
    V^\top P V = V^\top P^\top V &\succ 0, \\ M^\top \left( F-G^\top W V^\top P V W^\top G + H^\top Q H \right)M &\succ 0. \label{eq:lem:ineq_ROM_M}
\end{align}
\end{subequations}

First, note that $V^\top P V= V^\top P^\top V \succ 0$ is satisfied because $P = P^\top \succ 0$ and $V$ is full column rank by assumption. 
Second, if
\begin{equation}
     \hat \LL_{\hat n} \coloneqq F-G^\top W V^\top P V W^\top G + H^\top Q H \succ 0
\end{equation}
is satisfied, then~\eqref{eq:lem:ineq_ROM_M} is also satisfied because $M$ is full column rank.
From here, it is only left to show that $\LL_n \succ 0 \Rightarrow \hat \LL_{\hat n} \succ 0$, i.e., 
\begin{multline}\label{eq:lem:ineq}
    F-G^\top W V^\top P V W^\top G + H^\top Q H \succeq \\
    F-G^\top P G + H^\top Q H \succ 0,
\end{multline}
since $F-G^\top P G + H^\top Q H \succ 0$ by assumption.
Inequality~\eqref{eq:lem:ineq} is implied if
\begin{equation}\label{eq:proof:ineq1}
    -W V^\top P V W^\top \succeq - P,
\end{equation}
holds. 
Substituting $W^\top = \left(V^\top P V \right)^{-1} V^\top P$ (as in~\eqref{eq:def_W}) in~\eqref{eq:proof:ineq1} leaves
\begin{align}
    \begin{split}
        -PV\left(V^\top P V\right)^{-1} V^\top P V \left(V^\top P V\right)^{-1} V^\top P &= \\
        -PV\left(V^\top P V\right)^{-1} V^\top P
        &\succeq -P.
    \end{split}
\end{align}
The latter inequality is true by Lemma~\ref{lem:P_ineq}.
Therefore, the reduced-order REN is robust with the same $(Q,S,R)$.
The proof is completed by noting that a REN is also contractive under the robustness conditions of Theorem~\ref{thm:LMI_cond}.
\end{proof}